\documentclass[11pt]{article}

\usepackage{graphicx}
\usepackage{amsmath}
\usepackage{amssymb}
\usepackage{amsthm}
\usepackage{hyperref}
\usepackage{fullpage}

\newtheorem{theorem}{Theorem}
\newtheorem{lemma}{Lemma}
\newtheorem{innercustomthm}{Theorem}
\newenvironment{customthm}[1]
  {\renewcommand\theinnercustomthm{#1}\innercustomthm}
  {\endinnercustomthm}
\newtheorem{proposition}{Proposition}
\newtheorem{corollary}{Corollary}

\newcommand\numberofauthors[1]{}
\newcommand\affaddr[1]{#1}
\newcommand\alignauthor{}
\let\email=\url

\def\etal{{et~al.}}
\newcommand{\NN}{\mathbb{N}}
\newcommand{\RR}{\mathbb{R}}

\begin{document}

\title{Diffuse Reflection Diameter in Simple Polygons\footnote{A preliminary version of this work has been published as: G. Barequet, S. M. Cannon, E. Fox-Epstein, B. Hescott, D. L. Souvaine, C. D. T\'{o}th, A. Winslow, Diffuse reflections in simple polygons, Electronic Notes in Discrete Mathematics {\bf 44(5)} (2013), 345--350.}}

\author{
\numberofauthors{7}
\alignauthor
Gill Barequet\\
\affaddr{Department of Computer Science}\\
\affaddr{Technion}\\
\affaddr{Haifa, Israel}\\
\email{barequet@cs.technion.ac.il}
\and
\alignauthor
Sarah M. Cannon\thanks{Supported in part by National Science Foundation grant CCF-0830734.}\\
\affaddr{College of Computing}\\
\affaddr{Georgia Institute of Technology}\\
\affaddr{Atlanta, GA}\\
\email{sarah.cannon@gatech.edu}
\and
\alignauthor
Eli Fox-Epstein\footnotemark[2]\\
\affaddr{Department of Computer Science}\\
\affaddr{Brown University}\\
\affaddr{Providence, RI}\\
\email{ef@cs.brown.edu}
\and
\alignauthor
Benjamin Hescott\\
\affaddr{Department of Computer Science}\\
\affaddr{Tufts University}\\
\affaddr{Medford, MA}\\
\email{hescott@cs.tufts.edu}
\and
\alignauthor
Diane L. Souvaine\footnotemark[2]\\
\affaddr{Department of Computer Science}\\
\affaddr{Tufts University}\\
\affaddr{Medford, MA}\\
\email{dls@cs.tufts.edu}
\and
\alignauthor
Csaba D. T\'oth\footnotemark[2]\\
\affaddr{Department of Mathematics}\\
\affaddr{California State University Northridge}\\
\affaddr{Los Angeles, CA}\\
\email{cdtoth@acm.org}
\and
\alignauthor
Andrew Winslow\footnotemark[2]\\
\affaddr{D\'{e}partment d'Informatique}\\
\affaddr{Universit\'{e} Libre de Bruxelles}\\
\affaddr{Brussels, Belgium}\\
\email{awinslow@ulb.ac.be}
}

\date{}
\maketitle

\begin{abstract}
We prove a conjecture of Aanjaneya, Bishnu, and Pal that the minimum number of \emph{diffuse reflections} sufficient to illuminate the interior of any simple polygon with $n$ walls from any interior point light source is $\lfloor n/2 \rfloor - 1$. Light reflecting diffusely leaves a surface in all directions, rather than at an identical angle as with specular reflections.
\end{abstract}

\section{Introduction}
\label{sec:intro}

For a light source placed in a polygonal room with mirror walls, light rays that reach a wall at angle $\theta$, with respect to the normal of the wall's surface, also leave at angle $\theta$.
In other words, for these \emph{specular reflections} the angle of incidence equals the angle of reflection (see Fig.~\ref{fig:two-reflections}).

\begin{figure}[hbt]

\centering
\includegraphics[width=0.5\textwidth]{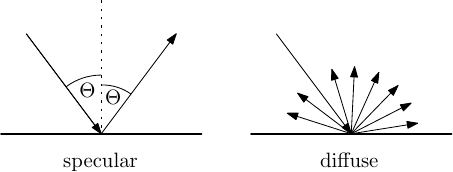}
\caption{Two types of reflections.
Specular reflection occurs on mirrored surfaces (left) and diffuse reflection occurs on matte surfaces (right).}
\label{fig:two-reflections}
\end{figure}

Klee~\cite{K69} asked whether the interior of any room defined by a simple polygon with mirrored walls is completely illuminated by placing a single point light anywhere in the interior. Tokarsky~\cite{T95} gave a negative answer to this question by constructing simple polygons and pairs of points $(s, t)$ such that there is no path from $s$ to $t$ with specular reflections off the walls of the room.

\begin{figure}[htb]
\centering
\includegraphics[width=.95\textwidth]{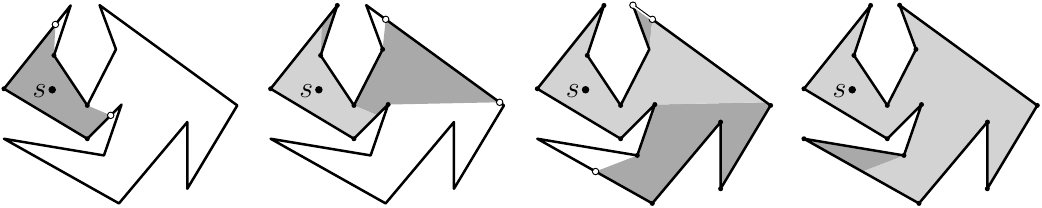}
\caption{The regions of the polygon illuminated by a light source $s$ after 0, 1, 2, and 3 diffuse reflections.}
\label{fig:diffuse-ex}
\end{figure}

On the other hand, if the walls of the polygonal room $P$ reflect light diffusely in all directions, then it is easy to see that every point in $P$ is illuminated after at most $n$ \emph{diffuse reflections} (Fig.~\ref{fig:diffuse-ex}). For diffuse reflections, we assume that the vertices
of $P$ absorb light, and that light does not propagate along the edges of $P$.
A \emph{diffuse reflection path} is a polygonal path $\gamma$ contained in $P$ such that
every interior vertex of $\gamma$ lies in the relative interior of some edge of $P$,
and the relative interior of every edge of $\gamma$ is in the interior of $P$.

Aronov \etal~\cite{AD+98} were the first to study $V_k(s)$, the part of the polygon illuminated by a light source $s$ after at most $k$ diffuse reflections. Formally, $V_k(s)$ is the set of points $t\in P$ such that there is a diffuse reflection path from $s$ to $t$ with at most $k$ interior vertices. In particular, $V_0(s)$ is the visibility region of point $s$ in the interior of $V$ (where the boundary of $P$ is considered opaque), hence it is a simply-connected region with $O(n)$ edges~\cite{AGS00}. Aronov \etal proved that $V_1(s)$ is simply connected with at most $\Theta(n^2)$ edges. Brahma \etal~\cite{BPS04} constructed simple polygons and a source $s$ such that $V_2(s)$ is not simply connected, and showed that $V_3(s)$ can have as many as $\Omega(n)$ holes.
Extending the work of~\cite{AD+98}, Aronov \etal~\cite{ADD+98,AD06} and Prasad \etal~\cite{PPD98} bounded the complexity of $V_k(s)$ at $O(n^9)$ and $\Omega(n^2)$ for all $k$.
It remains an open problem to close the gap between these bounds for $k \geq 2$.

Finding a shortest diffuse reflection path between two given points in a simple polygon by brute force is possible in $O(n^{10})$ time using the result of Aronov \etal~\cite{AD06}.
Ghosh \etal~\cite{GG+09} presented a 3-approximation in $O(n^2)$ time, and their
approximation applies even if the polygon $P$ has holes.

\begin{figure}[htb]
\centering
\includegraphics[width=0.6\textwidth]{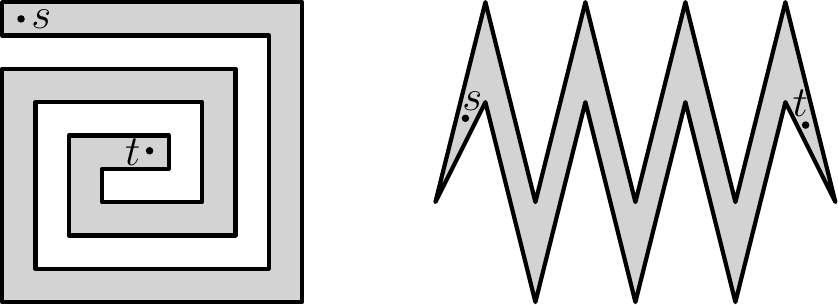}
\caption{Left: An orthogonal spiral polygon with $n=20$ vertices~\cite{ABP08}, where every diffuse
reflection path between $s$ and $t$ has at least $\lceil n/2 \rceil - 2 = 8$ turns.
Right: A zig-zag polygon with $n=16$ vertices where every diffuse reflection path between $s$ and $t$ has at least $\lfloor n/2 \rfloor - 1=7$ reflections.}
\label{fig:spiral-snake}
\end{figure}

\paragraph{Results}
We determine the minimum number of diffuse reflections sufficient to illuminate the interior of any simple polygon with $n$ vertices from any interior point $s$. For a simple polygon $P$, the \emph{diffuse reflection diameter} $D(P)$ is the smallest $k\in \mathbb{N}_0$ such that for every two  points $s,t\in {\rm int}(P)$, there is a diffuse reflection path between $s$ and $t$ with at most $k$ interior vertices (i.e., with at most~$k$ reflections). For an integer $n\geq 3$, let $D(n)$ be the largest diffuse reflection diameter $D(P)$ over all simple polygons $P$ with $n$ vertices. Aanjaneya \etal~\cite{ABP08} conjectured that $D(n) \leq \lceil n/2 \rceil - 1$ and constructed a family of polygons that yields $D(n) \geq \lfloor n/2 \rfloor - 2$; see Fig.~\ref{fig:spiral-snake}~(left).
The family of zig-zag polygons (Fig.~\ref{fig:spiral-snake}, right) shows that $D(n) \geq \lfloor n/2 \rfloor - 1$ for all $n\geq 3$. Here we prove that this bound is tight.

\begin{theorem}\label{thm:upper}
We have $D(n)= \lfloor n/2\rfloor - 1$ for every integer $n\geq 3$.
\end{theorem}

When the points $s$ and $t$ are allowed to be on the boundary of $P$, the minimum number of diffuse reflections may be larger, since a diffuse reflection path cannot have edges along the boundary of $P$. Similarly to $D(P)$, we define $\overline{D}(P)$ as the smallest $k\in \mathbb{N}_0$ such that for every two points $s,t\in P$ (in the interior or on the boundary of $P$), there is a diffuse reflection path between $s$ and $t$ with at most $k$ interior vertices. For $n\geq 3$, let $\overline{D}(n)$ be the maximum $\overline{D}(P)$ over all simple polygons $P$ with $n$ vertices. We determine $\overline{D}(n)$ for all $n\geq 3$.

\begin{theorem}\label{thm:boundary}
We have $\overline{D}(3)=2$ and $\overline{D}(n)= \lfloor n/2\rfloor$ for every integer $n\geq 4$.
\end{theorem}

\paragraph{Related Results for Link Paths}
The diffuse reflection path is a special case of a \emph{link path}, which has been studied extensively due to its applications in motion planning, robotics, and curve compression~\cite{G07,MSD00}.
The \emph{link distance} between two points, $s$ and $t$, in a simple polygon $P$ is the minimum number of edges in a polygonal path between $s$ and $t$ that lies entirely in $P$. In a polygon $P$ with $n$ vertices, the link distance between two points can be computed in $O(n)$ time~\cite{S86}. The \emph{link diameter} of $P$, the maximum link distance between any two points in $P$, can be computed in $O(n\log n)$ time~\cite{S90}. By contrast, no polynomial time algorithm is known for computing the diffuse reflection diameter of a simple polygon.

\section{Preliminary Definitions}

For a set $S \subseteq\RR^2$, let ${\rm int}(S)$ and ${\rm cl}(S)$ denote the interior of $S$ and the closure of $S$, respectively. The boundary of $S$, denoted $\partial S$, is ${\rm cl}(S)\setminus {\rm int}(S)$. The relative interior of a line segment $pq$ in the plane is denoted ${\rm relint}(pq)$. Let $d(p, q)$ be the Euclidean distance between points $p$ and $q$ in the plane.

Let $P$ be a simple closed polygonal domain (for short, \emph{simple polygon}) with $n$ vertices, where  $n\geq 3$.
We say that two points $s,t\in P$ \emph{see} each other (or, are \emph{visible} to each other) if
${\rm relint}(st)\subset {\rm int}(P)$. In particular, consecutive vertices of a diffuse reflection path see each other.\footnote{Note that a more relaxed definition of visibility, that requires only $st\subset P$,
is common in the literature~\cite{AGS00}.}

A \emph{chord} of $P$ is a closed line segment $ab$, such that $a,b\in \partial P$ and ${\rm relint}(ab)\subset {\rm int}(P)$. Two line segments (e.g., chords of $P$) \emph{cross} each other if there is a point in the relative interior of both segments, but the two segments are not collinear. We define the \emph{visibility polygon} of a line segment $ab$ of $P$, denoted $V_0(ab)$, as the set of points visible from some point in ${\rm relint}(ab)$. ($V_0(ab)$ is also known as the \emph{weak visibility polygon} of the relative interior of $ab$~\cite{AGS00}.) A subset $U$ of $P$ \emph{weakly covers} an edge $e$ of $P$ if $U$ intersects ${\rm relint}(e)$.

\section{A Sequence of Regions $R_k$}
\label{sec:regions}

Let $P$ be a simple polygon with $n$ vertices, and let $s \in P$. Instead of tackling $V_k(s)$ directly, we recursively define an infinite sequence of simply-connected regions $R_0 \subseteq R_1 \subseteq R_2 \subseteq \ldots$ such that $R_0=V_0(s)$ and $R_k \subseteq V_k(s)$ for all $k\in \NN$. In Section~\ref{sec:upper}, we prove ${\rm int}(P)\subseteq R_{\lfloor n/2 \rfloor - 1}$ for all $s\in {\rm int}(P)$ and $n\geq 3$, which immediately implies Theorem~\ref{thm:upper}. In Section~\ref{thm:boundary}, we prove $P\subseteq R_{\lfloor n/2 \rfloor}$ for all $s\in P$ and $n\geq 4$, which implies Theorem~\ref{thm:boundary}.

Let $R_0=V_0(s)$. In the remainder of this section, we recall a few well-known characteristics of $V_0(s)$, and then formulate properties (i)--(iv) that we wish to maintain for all $R_k$, $k\in \NN_0$. Using (i)--(iv), we define $R_k$, $k\in \NN$, recursively, and show that (i)--(iv) are maintained in each step. Finally, we prove $R_k(s)\subseteq V_k(s)$ for all $k\in \NN_k$.

\begin{figure}[htb]
\centering
\includegraphics[width=.95\textwidth]{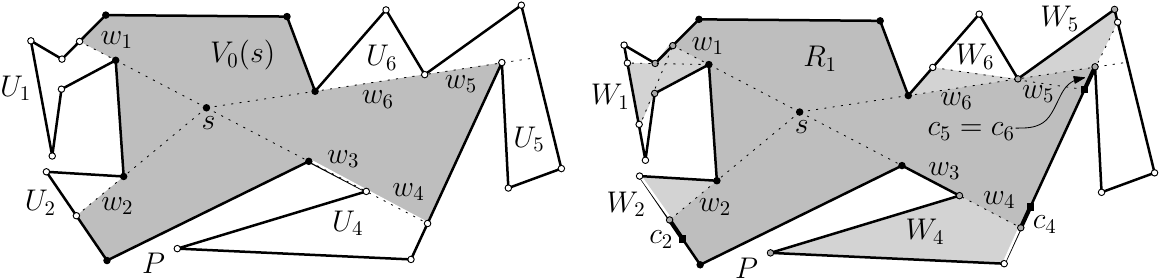}
\caption{Left: A simple polygon $P$ with a light source $s$.
The visibility polygon $V_0(s)$ has six windows: $w_1,\ldots , w_6$.
Window $w_3$ is the only degenerate window. Windows $w_1$ and $w_3$ are saturated, while the other windows are unsaturated.
Right: $R_1$ is the union of the closure of $V_0(s)$ and the
visibility polygons $W_{w_i}$ for $i=1,\ldots ,6$.}
\label{fig:preliminary}
\end{figure}

\paragraph{Properties of $V_0(s)$}
Recall that $V_0(s)$ is the set of all points $t\in P$ such that ${\rm relint}(st)\subseteq {\rm int}(P)$. Refer to Fig.~\ref{fig:preliminary}. As such, $V_0(s)$ is the union of (an infinite number of) closed line segments, each of which is incident to $s$ and some point in $\partial P$, hence $V_0(s)$ is simply connected. Consequently, the boundary of $V_0(s)$ consists of some line segments along $\partial P$ and possibly segments lying on rays emitted by $s$, which may contain a chord of $P$. However, $ V_0(s) \cap \partial P$ contains only one point along each ray emitted by $s$. For a set $U$, where $U\subseteq P$, we define a \emph{window} of $U$ to be a chord of $P$ contained in $\partial U$, or an edge $e$ of $P$ such that $e\subset \partial U$ but ${\rm relint}(e)\not\subset U$. A window which is an edge of $P$ is called a \emph{degenerate} window. See Fig.~\ref{fig:preliminary} for examples.

If the vertices of $P$ and $s$ are in general position (that is, no three points in a line), then every ray emitted by $s$ contains at most one window $V_0(s)$, and no window is degenerate. In general, however, a ray may contain several collinear windows $V_0(s)$, some of which may be degenerate (Fig.~\ref{fig:preliminary}). Suppose that $ab$ is a window of $V_0(s)$ such that $a$ lies in the interior of $sb$. For a window $ab$ of $V_0(s)$, consider the maximal line segment $a'b'$ such that $ab\subseteq a'b'$ and $a'b'\subseteq \partial V_0(s)$ (possibly, $ab=a'b'$). Then the rays emitted by $s$ can reach $\partial P$ in a neighborhood of $b'$, and a sufficiently small neighborhood contains a segment $c\subseteq \partial P$, where ${\rm relint}(c)$ is visible from the other endpoint $a$ of the window. The windows of $V_0(s)$ that lie on distinct rays are necessarily disjoint.
Consequently, the region $R_0=V_0(s)$ satisfies the following properties.

\begin{enumerate}\itemsep 0pt
\item[{\rm \bf (i)}] The closure of $R_k$, ${\rm cl}(R_k)$, is a simple polygon whose boundary consists of chords of $P$ and line segments contained in the boundary of $P$.
\item[{\rm \bf (ii)}] The endpoints of every window of $R_k$ can be labeled by $a$ and $b$ such that ${\rm cl}(R_k)$ has an interior angle of at least $180^\circ$ at $a$.
\item[{\rm\bf (iii)}] For every window $ab$ of $R_k$, there is a nontrivial line segment $c$ such that an endpoint of $c$ lies on the ray $\overrightarrow{ab}$, ${\rm relint}(c)$ lies in the relative interior of an edge of $P$, ${\rm relint}(c)\subset R_k$, and $a$ sees every point in ${\rm relint}(c)$.
\item[{\rm \bf(iv)}] Any two windows of $R_k$ are disjoint or collinear; and for any two adjacent windows, there is a common segment $c$ that satisfies property (iii).
\end{enumerate}

\paragraph{Recursive construction of $R_k$, $k\geq 1$}
We can now construct $R_{k+1}$ for all $k\in \NN_0$, assuming that $R_k$ is already defined and satisfies (i)--(iv). Intuitively, we construct $R_{k+1}$ by extending $R_k$ beyond each of its windows with a visibility region as follows (refer to Fig.~\ref{fig:preliminary}). A window $ab$ of $R_k$ is \emph{saturated} if every chord of $P$ that crosses $ab$ has an endpoint in $R_k$; otherwise, it is \emph{unsaturated}. Note that every degenerate window is saturated, because a degenerate window crosses no chords.

Each nondegenerate window $ab$ of $R_k$ decomposes $P$ into two simple polygons; let $U_{ab}$ denote the polygon that is disjoint from ${\rm int}(R_k)$. For a degenerate window, let $U_{ab}=ab$. For each window $ab$, we define a set $W_{ab}$ as follows.
If $ab$ is saturated, then let $W_{ab}=V_0(ab)\cap U_{ab}$.
If $ab$ is unsaturated, then let $c\subset R_k\cap \partial P$ be the segment described in property (iii), and let $W_{ab}=V_0(c)\cap U_{ab}$. Let $R_{k+1}$ be the union of ${\rm cl}(R_k)$ and the sets $W_{ab}$ for all windows $ab$ of $R_k$. The definition of the regions $R_k$, $k\in \NN_0$, readily implies that properties (i)--(iv) are maintained for $R_{k+1}$.

\begin{proposition}
\label{pp:property}
Let $P$ be a simple polygon and $s\in P$. For every $k\in \mathbb{N}_0$, region~$R_k$ satisfies properties (i)--(iv).
\end{proposition}
\begin{proof}
We proceed by induction on $k\in \mathbb{N}_0$. For $k=0$, the region $R_0$ is the visibility polygon $V_0(s)$ of point $s$ in $P$, and properties (i)--(iv) are easily verified (see Figure~\ref{fig:preliminary}, left).
Suppose $R_k$ satisfies (i)--(iv) for some $k\in \mathbb{N}_0$. If $R_k$ has no window, then
${\rm cl}(R_k)=P$ and $R_{k+1}=P$, hence properties (i)-(iv) trivially hold for $R_{k+1}$.
If $R_k$ has at least one window, then $R_{k+1}$ is the union of ${\rm cl}(R_k)$ and the visibility polygons $W_{ab}$ for all windows $ab$. By definition, $W_{ab}$ contains ${\rm relint}(ab)$ for both saturated and unsaturated window $ab$. Each $W_{ab}$ satisfies properties (i)--(iv) within $U_{ab}$. This proves properties (i)--(iii) for $R_{k+1}$, and (iv) for windows adjacent in each $W_{ab}$.

It remains to establish (iv) for pairs of windows, $w$ and $w'$, of $R_{k+1}$ that lie on the boundary of $W_{ab}$ and $W_{a'b'}$, where $ab$ and $a'b'$ are distinct windows of $R_k$. Suppose that their common endpoint is $x=w\cap w'$. Then $x$ is also a common endpoint of $ab$ and $a'b'$. Since $R_k$ satisfies (iv) by the induction hypothesis, the windows $ab$ and $a'b'$ are collinear, and they have a common segment $c$ satisfying (iii). Consequently, $w$ and $w'$ lie on the same side of $ab\cup a'b'$. Note that $ab$ is unsaturated, otherwise $W_{ab}$ would weakly cover the edge of $U_{ab}$ incident to $x$, and $w$ could not be incident to $x$. Analogously, $a'b'$ is unsaturated. However, if both $ab$ and $a'b'$ are unsaturated, then $V_0(c)$ weakly covers the edge of $U_{ab}$ or $U_{a'b'}$ incident to $x$. Therefore, at most one of $w$ and $w'$ can be incident to $x$. We conclude that the windows $w$ and $w'$ of $R_{k+1}$ are disjoint, proving property (iv) for $R_{k+1}$.
\end{proof}

The next proposition justifies that the closure of $R_k$ is contained in $V_{k+1}(s)$ if $R_k\subseteq V_k(s)$ and $s\in {\rm int}(P)$.

\begin{proposition}
\label{pp:closure}
Let $s\in {\rm int}(P)$ and $k\in \mathbb{N}_0$.
For every set $U\subseteq V_k(s)$, we have ${\rm cl}(U)\subseteq V_{k+1}(s)$.
\end{proposition}

\begin{proof}
Let $p\in \partial U\setminus U$.
Since $p\in {\rm cl}(U)$ and ${\rm cl}(U)\subseteq {\rm cl}(V_k(s))$, we have $s\in {\rm cl}(V_k(p))$ by symmetry. For every $i\in \mathbb{N}$, there is a point $s_i \in V_k(p)$ lying in a $\frac{1}{i}$-neighborhood of $s$ such that there is a diffuse reflection path $(p,r_i(1),\ldots , r_i(\ell),s_i)$ with $\ell\leq k$, where the points $r_i(1),\ldots , r_i(\ell)$ lie in the interior of some edges of $P$.

By construction, we have $\lim_{i \rightarrow \infty}d(s_i, s) = 0$, and we may assume by compactness that there is a point $r\in \partial P$ such that $\lim_{i \rightarrow \infty}d(r_i(\ell), r) = 0$. The ray $\overrightarrow{r_i(\ell)s_i}$ hits $\partial P$ at a point $q_i$, and we may assume that there is a point $q\in \partial P$ such that $\lim_{i \rightarrow \infty}d(q_i(\ell), q) = 0$, where $s$ lies on the chord $rq$. For a sufficiently large $i\in \mathbb{N}$, there is a point $q'\in \partial P$ in a neighborhood of $q$ that lies in the interior of some edge of $P$ and directly sees both $r_i(\ell)$ and $s$. Consequently, there is a diffuse reflection path $(p,r_i(1),\ldots , r_i(\ell),q',s)$ of length at most $k+1$ between $p$ and $s$. It follows that $p \in V_{k+1}(S)$, and so  ${\rm cl}(U)\subseteq V_{k+1}(s)$ as desired.
\end{proof}

\begin{corollary} \label{cor:rk-in-vk}
If $s\in {\rm int}(P)$, then $R_k\subseteq V_k(s)$ for all $k\in \mathbb{N}_0$.
\end{corollary}
\begin{proof}
We prove the statement by induction on $k$.
In the base case we have $R_0 = V_0(s)$ by definition.
Suppose $R_k\subseteq V_k(s)$ for some $k\in \mathbb{N}_0$. By Proposition~\ref{pp:closure}, ${\rm cl}(R_k)\subseteq V_{k+1}(s)$. For every window $ab$ of $R_k$, we show that $W_{ab}\subseteq V_{k+1}$. Specifically, consider the two cases in the construction of $W_{ab}$.

First, suppose $ab$ is a saturated window of $R_k$. Then every point $t\in W_{ab}$ sees some point
$x\in {\rm relint}(ab)$, and $tx$ is contained in a chord $ty$ of $P$, where $y\in R_k$. If $y$ is in the relative interior of an edge of $P$, then a diffuse reflection path from $s$ to $y$ can be extended to $t$ using a diffuse reflection at $y$. Otherwise, note that $t$ also sees some neighborhood of $x$ within
${\rm relint}(ab)$, hence some neighborhood of $y$ within $\partial P\cap R_k$. Again, a diffuse reflection path from $s$ to such a point can be extended to $t$.

Now suppose that $ab$ is unsaturated. Then every point $t\in W_{ab}$ sees a point in the relative interior of segment $c$, where ${\rm relint}(c)\subseteq R_k$ and ${\rm relint}(c)$ lies in the relative interior of an edge of $P$. A diffuse reflection path from $s$ to any point in $c$ can be extended to $t$ via a diffuse reflection in $c$. In both cases, we have shown $W_{ab}\subset V_{k+1}(s)$.
Consequently, $R_{k+1}\subseteq V_{k+1}(s)$.
\end{proof}

\paragraph{Weakly covered edges}
We associate two crucial parameters with the regions $R_k$, $k\in \NN_0$. For every $k\in \NN_0$, let $\mu_k$ be the number of edges of $P$ weakly covered by $R_k$, and $\lambda_k$ the total number of windows of $R_k$. We derive a lower bound on the number of new edges weakly covered in each round.

\begin{lemma}\label{lem:saturated}
For every $k\in \NN_0$,
\begin{itemize}\itemsep -1pt
\item[{\rm (1)}] We have $\mu_{k+1}\geq \mu_k+\lambda_k$; and
\item[{\rm (2)}] If all windows of $R_k$ are saturated, then $\mu_{k+1}\geq \min(\mu_k+\lambda_k+1,n)$.
\end{itemize}
\end{lemma}
\begin{proof}
Recall that ${\rm cl}(R_k)\subseteq R_{k+1}$, and so $R_{k+1}$ contains all degenerate windows of $R_k$. Now consider nondegenerate windows of $R_k$.

Let $ab$ be a nondegenerate window of $R_k$.
By property (ii), $a$ is a flat or reflex vertex of ${\rm cl}(R_k)$, hence it is a convex vertex of $U_{ab}$. Let $ad$ denote the edge of $P$ incident to $a$ and on the boundary of $U_{ab}$. It is clear that $R_k$ does not weakly cover $ad$, and we show that $W_{ab}$ weakly covers it.
If $ab$ is saturated, it is clear that $W_{ab}$ weakly covers $ad$. If $ab$ is unsaturated,
then $U_{ab}$ and $c$ lie on opposite sides of the line spanned by $ab$, and so every point
in $c$ sees some part of $ad$ in a neighborhood of $a$. Consequently, $R_{k+1}$ weakly covers at least one new edge of $P$ behind every window of $R_k$.

For the second claim, assume that all windows of $R_k$ are saturated, but $\mu_{k+1}<n$. Then there is a saturated window $ab$ such that $R_{k+1}$ does not weakly cover all edges of $P$ in $U_{ab}$. As above, let $ad$ denote the edge of $U_{ab}$ incident to $a$, and also let $e$ denote the edge of $P$ that contains $b$ and has nontrivial intersection with the boundary of $U_{ab}$. From above, we know that $R_{k+1}$ weakly covers $ad$. Next, consider all chords of $P$ that cross $ab$ and are parallel to $ad$ or $e$. At least one of these chords has an endpoint in the relative interior of some edge of $P$ that is disjoint from $R_k$ and is not $ad$, and so $W_{ab}$ weakly covers at least two new edges of $P$ behind $ab$, as required.
\end{proof}
\begin{corollary}\label{cor:saturated}
For every $k\in \NN_0$,
\begin{itemize}\itemsep -1pt
\item[{\rm (1)}] We have $\mu_{k+1}\geq \min(\mu_k+1,n)$; and
\item[{\rm (2)}] If all windows of $R_k$ are saturated, then $\mu_{k+1}\geq \min(\mu_k+2,n)$.
\end{itemize}
\end{corollary}
\begin{proof}
Note that if ${\rm cl}(R_k)\neq P$, then $R_k$ has at least one window and $\lambda_k\geq 1$.
\end{proof}

\section{Counting Weakly Covered Edges in $R_k$}
\label{sec:upper}

Let $P$ be a simple polygon with $n$ vertices, and let $s\in {\rm int}(P)$. In this section, we establish
the inequality
\begin{equation}\label{eq:1}\tag{$\star$}
\mu_k\geq \min(2k+3,n) 
\end{equation}
for all $k \in \NN_0$, which immediately implies Theorem~\ref{thm:upper}. It is folklore that $V_0(s)$ weakly covers at least three edges, hence $\mu_0\geq 3$.
\begin{proposition}\label{pp:ini}
If $s\in {\rm int}(P)$, then $V_0(s)$ weakly covers at least three edges of $P$. Consequently, $\mu_0\geq 3$.
\end{proposition}
\begin{proof}
In any triangulation of $P$, $s$ lies in some triangle whose vertices partition the edges of $\partial P$ into three sets. At least one edge is seen by $s$ in each of the three sets.
\end{proof}

We prove \eqref{eq:1} for all $k\in \NN_0$ by induction on $k$. Recall that $R_0$ satisfies \eqref{eq:1} by Proposition~\ref{pp:ini}, and $\mu_k$ strictly monotonically increases until it reaches $n$ by Corollary~\ref{cor:saturated}. Consequently, if \eqref{eq:1} fails for some $R_{k+1}$, $k\in \NN_0$, then $R_k$ must satisfy \eqref{eq:1} with equality, and $\mu_k<n$. This motivates the following definition.
A region $R_k$ is called \emph{critical} if $\mu_k=2k+3$ and $\mu_k<n$.

By Lemma~\ref{lem:saturated}, it is enough to show that whenever $R_k$ is critical, then $\lambda_k\geq 2$ or all windows of $R_k$ are saturated. For every critical region $R_k$, we will inductively show (Lemma~\ref{lem:induction}(3)) 
that one of the following two conditions holds:
\begin{enumerate}\itemsep -1pt
\item[{\rm \bf (A)}] All windows of $R_k$ are saturated;  or
\item[{\rm \bf (B)}] $R_k$ has an unsaturated window and $\lambda_k \geq 2$.
\end{enumerate}
Note that these two conditions are mutually exclusive, that is, a region $R_k$ cannot satisfy both.

\paragraph{Initialization}
We first show that $R_0$ satisfies one of the two conditions.

\begin{proposition}\label{pp:ini1}
Region $R_0$ satisfies condition~(A) or~(B).
\end{proposition}
\begin{proof}
First, suppose $\lambda_0 \geq 2$. Then either all windows are saturated so (A) holds, or at least one window is unsaturated and (B) holds.

Next, suppose $\lambda_0=1$. If the single window of $R_0$ is degenerate, then (A) holds. Assume that $R_0$ has exactly one window that is nondegenerate, denoted by $ab$ as defined in (ii). Recall that $R_0=V_0(s)$, and so every point in $\partial R_0$ is contained in a window or directly visible from $s$. Consequently, all points in $\partial R_0\setminus ab$ are in $R_0$. As $ab$ splits $P$ into $R_0$ and $U_{ab}$, every chord of $P$ that crosses $ab$ has exactly one endpoint in $\partial R_0\setminus ab$ and so has one endpoint in $R_0$, as desired. It follows that window $ab$ is saturated, and so $R_0$ satisfies condition~(A).
\end{proof}

We will also inductively show (Lemma~\ref{lem:induction}) that no two consecutive critical regions satisfy (B). For the first two regions, $R_0$ and $R_1$, this is established as follows.

\begin{proposition}\label{pp:ini2}
If both $R_0$ and $R_1$ are critical and $R_0$ satisfies (B), then $R_1$ satisfies (A).
\end{proposition}
\begin{proof}
If both $R_0$ and $R_1$ are critical, then $\mu_0=3$ and $\mu_1=5$. By the proof of Proposition~\ref{pp:ini1}, $R_0$ must have at least two disjoint windows, otherwise it satisfies (A) instead of (B). Additionally $R_0$ must have at most two windows, as otherwise $\lambda_0 \geq 3$ and so $\mu_1 > 5$ by Lemma~\ref{lem:saturated}(1), a contradiction.

Denote by $a_1b_1$ and $a_2b_2$ the two disjoint windows of $R_0$. By assumption at least one of these two disjoint windows is unsaturated, and it follows that the other must be either unsaturated or saturated and degenerate.
The boundary of $R_0$ consists of five line segments: window $a_1b_1$, window $a_2b_2$, and three segments along three edges of $P$ weakly covered by $R_0$.

Since $a_1b_1$ and $a_2b_2$ are disjoint and thus not adjacent, they are both incident to some edge $e$ of $P$ weakly covered by $R_0$. By Property~(ii), we may assume that $a_1$ and $a_2$ are reflex vertices of ${\rm cl}(R_0)$. If $a_1$ or $a_2$ is incident to $e$, then both windows are saturated, contradicting our earlier observation. Therefore, neither $a_1$ nor $a_2$ is incident to $e$, hence both $b_1$ and $b_2$ lie on $e$.
Then $R_1={\rm cl}(R_0)\cup V_0(c_1)\cup V_0(c_2)$, where $c_1,c_2\subseteq e$. Since $R_1$ weakly covers precisely one new edge of $P$ behind each window, every window of $R_1$ is collinear with $e$. It follows that a chord of $P$ that crosses any window of $R_1$ cannot have an endpoint on edge $e$. Therefore one endpoint of such a chord is in $R_1={\rm cl}(R_0)\cup V_0(c_1)\cup V_0(c_2)$, and all windows of $R_1$ are saturated. As desired, $R_1$ satisfies (A).
\end{proof}

The next proposition explores the case where $R_{k+1}$ is critical for some $k\geq 1$, but the previous region $R_{k}$ satisfies \eqref{eq:1} with a strict inequality so is not critical.

\begin{proposition}\label{pp:crit}
Suppose that $R_k$ is not critical and $\mu_k> 2k+3$, but $R_{k+1}$ is critical.
Then $\lambda_k =
  1$, region $R_k$ has an unsaturated window,
and $R_{k+1}$ satisfies (A).
\end{proposition}

\begin{proof}
Because $\mu_k \geq 2k+4$ and $\mu_{k+1} = 2(k+1)+3 = 2k+5$, we have $\mu_{k+1}\leq \mu_k+1$. By Corollary~\ref{cor:saturated}(1), we have $\mu_k = 2k+4$ and $\lambda_k=1$, that is, $R_k$ has exactly one window.

\begin{figure}[htb]
\centering
\includegraphics[width=0.95\textwidth]{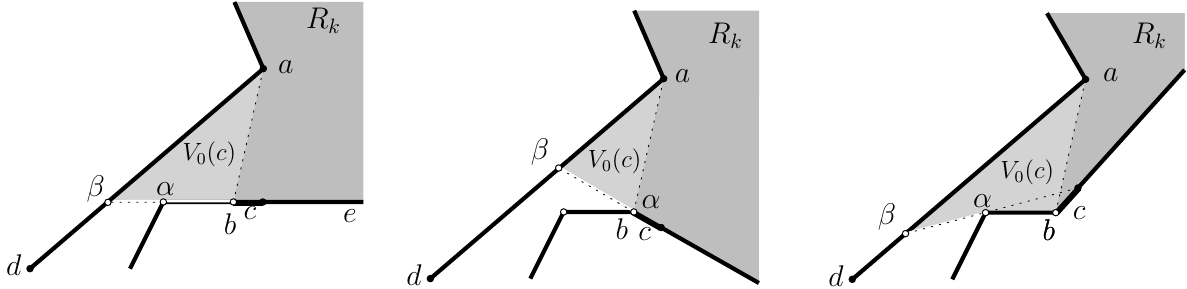}
\caption{The situation in Proposition~\ref{pp:crit}.
The region $R_k$ is noncritical and $R_{k+1}$ is critical. The window $ab$ of $R_k$ is unsaturated,
while the  window $\alpha\beta$ of $R_{k+1}$ is saturated.
Left: $b$ lies in the relative interior of an edge of $P$.
Middle: $b$ is a reflex vertex of $P$.
Right: $b$ is a convex vertex of $P$.}
\label{fig:prop-crit}
\end{figure}

The window $ab$ cannot be saturated by Corollary~\ref{cor:saturated}(2). Since $R_{k+1}$ is critical, we have $\mu_{k+1} < n$ and thus $R_{k+1}$ also has at least one window, which must be within $U_{ab}$.
Because $ab$ is unsaturated, $R_{k+1}={\rm cl}(R_k)\cup W_{ab}$, where $W_{ab}=V_0(c)\cap U_{ab}$ for a segment $c$ described in (iii). Since $\mu_{k+1}=\mu_k+1$, the region $R_{k+1}$ weakly covers precisely one more edge than $R_k$.
Let $ad$ be the edge of $P$ incident to $a$ lying on the boundary of $U_{ab}$. Refer to Fig.~\ref{fig:prop-crit}. As argued in the proof of Lemma~\ref{lem:saturated}, $ad$ is the only edge weakly covered by $W_{ab}$ but not weakly covered by $R_k$. We distinguish between two cases to define a point $\beta\in ad$.

{\bf Case~1: $b\in {\rm relint}(e)$ for some edge $e$ of $P$} (Fig.~\ref{fig:prop-crit}, left).
Then $c\subset {\rm relint}(e)$. Since $ad$ is the only edge in $U_{ab}$ visible from $c$,
the supporting line of $e$ intersects $ad$, and we denote the intersection point by $\beta$.

{\bf Case~2: $b$ is a vertex of $P$.} In this case, $b$ must be a reflex vertex of $P$ (as in Fig.~\ref{fig:prop-crit}, middle), otherwise $c$ would also see the edge of $U_{ab}$ incident to $b$ (as in Fig.~\ref{fig:prop-crit}, right). Since $ad$ is the only edge in $U_{ab}$ visible from $c$,
the supporting line of $c$ intersects $ad$, and we denote the intersection point by $\beta$.

In both cases, we have $W_{ab}=\Delta(ab\beta)\setminus b\beta$, and so any window of $R_{k+1}$ is contained in $b\beta$. Every chord of $P$ that crosses $b\beta$ has an endpoint in either ${\rm relint}(a\beta)$ or in ${\rm cl}(R_k)$. In either case, one endpoint of such a chord is in $R_{k+1}$, and so all windows of $R_{k+1}$ are saturated.
\end{proof}

\paragraph{Induction Step}
The next three propositions concern the situation where several consecutive regions are critical.

\begin{proposition}\label{pp:1}
If all windows of both $R_k$ and $R_{k+1}$ are collinear, and $R_k$ satisfies (A), then $R_{k+1}$ also satisfies (A).
\end{proposition}

\begin{proof}
Let $\alpha\beta$ be an arbitrary window of $R_{k+1}$. Then $\alpha\beta$ lies on the boundary of some visibility region $W_{ab}$, where $ab$ is a window of $R_k$. Since $ab$ is saturated, we have $W_{ab}=V_0(ab)\cap U_{ab}$.
Consider a chord $cd$ of $P$ that crosses $\alpha\beta$ with $d\in U_{\alpha\beta}$.
We need to show that $c\in R_{k+1}$. If $cd$ crosses $ab$, then $c\in R_k\subseteq R_{k+1}$ since $ab$ is saturated. If $c=a$ or $c=b$, then $c\in R_{k+1}$ since $a,b\in {\rm cl}(R_k)\subseteq R_{k+1}$. Otherwise, $c\in \partial W_{ab}\setminus (ab\cup \alpha\beta)$. Since $ab$ is saturated, all points of $\partial W_{ab}\setminus W_{ab}$ are in windows of $R_{k+1}$, which are collinear with $\alpha\beta$, so $c$ must lie in $W_{ab}$. In all cases, $c\in R_{k+1}$, and so the window $\alpha\beta$ is saturated.
\end{proof}

\begin{proposition} \label{pp:A}
If $R_k$ and $R_{k+1}$ are critical and $R_k$ satisfies~(A), then $R_{k+1}$ satisfies (A) or (B).
\end{proposition}
\begin{proof}
First, note that $\mu_k = 2k + 3$ and $\mu_{k+1} = 2k+5$ by criticality. Note that $R_k$ has at least one window, otherwise $\mu_{k+1}=n$, contradicting the criticality of $R_{k+1}$. If $\lambda_k\geq 2$, then by Lemma~\ref{lem:saturated}(2), $\mu_{k+1} \geq (2k + 3) + 2 + 1 > 2k + 5$, which is a contradiction. It follows that $R_k$ has exactly one window that is nondegenerate.

If $\lambda_{k+1}\geq 2$, then $R_{k+1}$ satisfies (A) or (B): either all windows are saturated and (A) holds, or it has an unsaturated window and (B) holds. If $\lambda_{k+1}\leq 1$, then any window of $R_{k+1}$ is saturated by Proposition~\ref{pp:1} and so (A) holds.
\end{proof}

\begin{proposition}\label{pp:prop-crit-three-step}
Suppose that $R_k$, $R_{k+1}$, and $R_{k+2}$ are critical, $R_k$ satisfies~(A), and $R_{k+1}$ satisfies~(B).
Then $R_{k+2}$ satisfies~(A).
\end{proposition}

\begin{proof}
Note that $\mu_k=2k+3$, $\mu_{k+1}=\mu_k+2$, and $\mu_{k+2}=\mu_{k+1}+2$.
By Lemma~\ref{lem:saturated}(2), $R_k$ has only one window, which is saturated; label this window $ab$ as described in Property (ii). Recall that in this case, $W_{ab} = V_0(ab)\cap U_{ab}$. We have $\lambda_{k+1}=2$, as condition (B) yields $\lambda_{k+1}\geq 2$ and Lemma~\ref{lem:saturated}(1) implies $\lambda_{k+1}\leq 2$. However, $R_{k+1}$ cannot have two adjacent collinear windows, otherwise both of these windows would be saturated by Proposition~\ref{pp:1}, and $R_{k+1}$ would satisfy (A). By property (iv), $R_{k+1}$ has exactly two windows that are disjoint.
We denote them $\alpha_1\beta_1$ and $\alpha_2\beta_2$, respectively.
Refer to Fig.~\ref{fig:prop-crit-three-step}.

By Property~(ii), we may assume that $a$ is a reflex vertex or a straight vertex of ${\rm cl}(R_k)$. Denote by $ad$ the edge of $P$ on the boundary of $U_{ab}$ and incident to $a$. Since $\mu_{k+1} = \mu_k+2$, region $R_{k+1}$ weakly covers precisely two new edges of $P$: one is $ad$ and call the other $e$.
Let $f$ be the edge of $P$ containing $b$ and weakly covered by $U_{ab}$. Note that $\alpha_1\beta_1$ and $\alpha_2\beta_2$ are disjoint and lie on the boundary of $W_{ab}$. Since $ad$ and $f$ can each be incident to at most one of them, we have $e\neq f$, and edge $e$ is incident to both $\alpha_1\beta_1$ and $\alpha_2\beta_2$.

The boundary of $W_{ab}$ is formed by segments $ab$, $\alpha_1\beta_1$, and $\alpha_2\beta_2$, and some part of the edges $ad$, $e$, and possibly $f$. Note that $ad\subset \partial W_{ab}$ since a point in $ab$ sufficiently close to $a$ can see all of $ad$. However, the intersection $f\cap \partial W_{ab}$ could be the single point $b$, or a nontrivial line segment connecting $b$ and an endpoint of $f$. Without loss of generality, we may assume that $\partial W_{ab}$ contains, in counterclockwise order, $\alpha_1\beta_1$, part of $e$,  $\alpha_2\beta_2$, $ad$, $ab$, and possibly part of $f$, but contains no other segments.

\begin{figure}[htb]
\centering
\includegraphics[width=0.9\textwidth]{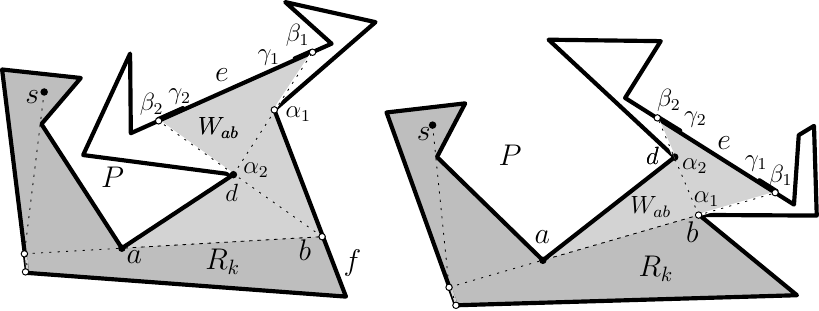}
\caption{The situation in Proposition~\ref{pp:prop-crit-three-step}. The region $R_k$ has a unique window $ab$, and $W_{ab}$ weakly covers two new edges of $P$: $ad$ and $e$. $R_{k+1}$ has two windows $\alpha_1\beta_1$ and $\alpha_2\beta_2$.
Left: $b$ lies in the relative interior of an edge $f$ of $P$.
Right: $b$ is a vertex of $P$ and $b=\alpha_1$.}
\label{fig:prop-crit-three-step}
\end{figure}

Since $\mu_{k+2}=\mu_{k+1}+2$, region $R_{k+2}$ weakly covers precisely one new edge of $P$ behind each of the two windows $\alpha_1\beta_1$ and $\alpha_2\beta_2$. It follows that $R_{k+2}$ has at most two windows: at most one behind each of $\alpha_1\beta_1$ and $\alpha_1\beta_2$.

By Property~(ii) and disjointness of $W_{ab}$'s two windows, we may assume that $\alpha_1$ and $\alpha_2$ are reflex vertices of ${\rm cl}(R_{k+1})$. Since $W_{ab}=V_0(ab)\cap U_{ab}$, the region $W_{ab}$ has a reflex or flat interior angle at both $\alpha_1$ and $\alpha_2$. We have $\alpha_1\in f$ (possibly $\alpha_1=b$), and $\alpha_2=d$.
Consequently, both $\beta_1$ and $\beta_2$ are contained in $e$. A segment in $P$ can connect two points in
${\rm relint}(\alpha_1\beta_1)$ and ${\rm relint}(\alpha_2\beta_2)$, respectively.
Therefore, $\alpha_1\beta_1$ and $\alpha_2\beta_2$ are unsaturated or degenerate windows.

Let $\gamma_1$ and $\gamma_2$, respectively, be the segments $c$ described in (iii) for the windows $\alpha_1\beta_1$ and $\alpha_2\beta_2$. Note that both $\gamma_1$ and $\gamma_2$ are in ${\rm relint}(e)$. By construction, $R_{k+2}={\rm cl}(R_{k+1})\cup V_0(\gamma_1)\cup V_0(\gamma_2)$. Since $R_{k+2}$ weakly covers exactly one new edge of $P$ behind each of $\alpha_1\beta_1$ and $\alpha_2\beta_2$, every window of $R_{k+2}$ is collinear with $e$. It follows that a chord of $P$ that crosses any window of $R_{k+2}$ cannot have an endpoint on the edge $e$, which contains the only uncovered portions of $\partial R_{k+2}$ that are on the boundary of $P$. Therefore one endpoint of such a chord is in $R_{k+2}={\rm cl}(R_{k+1})\cup V_0(\gamma_1)\cup V_0(\gamma_2)$.
Consequently, any window of $R_{k+2}$ is saturated.
\end{proof}

We are now in position to prove Lemma~\ref{lem:induction}. We note that claim~(1) of the lemma is the statement we want to prove, and we are able to do this by establishing a stronger induction argument also maintaining claims~(2) and~(3).
\begin{lemma}\label{lem:induction}
For all $k \in \mathbb{N}_0$,
\begin{itemize}\itemsep 0pt
\item[{\rm (1)}] $\mu_{k} \geq \min (2k + 3, n)$;
\item[{\rm (2)}] If $R_k$ is critical, it satisfies (A) or (B); and
\item[{\rm (3)}] If $R_k$ is critical and satisfies (B), then either $R_{k-1}$ is critical and satisfies (A), or $k=0$.
\end{itemize}
\end{lemma}
\begin{proof}
First, suppose $k = 0$. Then $\mu_0 \geq 3$ by Proposition~\ref{pp:ini}, satisfying (1). By Proposition~\ref{pp:ini1}, $R_0$ satisfies (A) or (B), proving (2). Claim (3) trivially holds for $k = 0$.


For the inductive step, suppose $k \geq  1$ and that (1), (2), and (3) hold for all smaller $k$. First, we establish (1). If $R_{k-1}$ is critical, then by the induction hypothesis it must satisfy (A) or (B). By criticality, we have $\mu_{k-1} = 2k+1$, and Lemma~\ref{lem:saturated} yields $\mu_k \geq \min(\mu_k + 2,n)= \min(2k+3,n)$. If $R_{k-1}$ is not critical, then $\mu_{k-1} \geq \min(2k+1,n)$ by the induction hypothesis and $\mu_{k-1} \neq 2k+1$ by the definition of criticality. Consequently $\mu_{k-1} \geq \min(2k+2,n)$, and Corollary~\ref{cor:saturated}(1) yields $\mu_k \geq \mu_{k-1} + 1 \geq \min(2k+3,n)$, proving (1).

To establish (2) and (3), suppose $R_k$ is critical. If $R_{k-1}$ is not critical, then $\mu_{k-1} \geq 2k+2$ from the discussion above. Applying Proposition~\ref{pp:crit} (for $k-1$ instead of $k$), it follows that $R_k$ satisfies (A). If $R_{k-1}$ is critical and satisfies (A), then $R_k$ satisfies (A) or (B) by Proposition~\ref{pp:A}.
It remains to consider the case that $R_{k-1}$ is critical and satisfies (B).

Suppose that both $R_k$ and $R_{k-1}$ are critical and $R_{k-1}$ satisfies (B).
Claim (3) implies (for $k-1$ instead of $k$) that either $k=1$ or $R_{k-2}$ is critical and satisfies (A). If $k=1$, then $R_k$ satisfies (A) by Proposition~\ref{pp:ini2}.
If $R_{k-2}$ is critical and satisfies (A), we apply Proposition~\ref{pp:prop-crit-three-step} (for $k-2$ instead of $k$) and conclude that $R_k$ satisfies (A). In all cases, $R_k$ satisfies (A) or (B), proving (2). If $R_k$ satisfies (B), then $R_{k-1}$ satisfies (A), proving (3).
\end{proof}

We can now finally prove Theorem~\ref{thm:upper}.

\begin{customthm}{\ref{thm:upper}}
We have $D(n)= \lfloor n/2\rfloor - 1$ for every integer $n\geq 3$.
\end{customthm}

\begin{proof}
We prove that in every simple polygon $P$ with $n\geq 3$ vertices,
there exists a diffuse reflection path with at most $\lfloor n/2\rfloor-1$
reflections between any two points $s,t\in {\rm int}(P)$.
It is enough to show that ${\rm int}(P) \subseteq V_k(s)$  for every
$s\in {\rm int}(P)$ and every $k\geq \lfloor n/2\rfloor-1$.

Note that $\lfloor n/2\rfloor -1= \lceil (n-3)/2\rceil$.
By Lemma~\ref{lem:induction}, $\mu_{\lfloor n/2\rfloor-1} \geq 2(\lfloor n/2\rfloor-1) + 3 \geq n$, so
$R_{\lfloor n/2\rfloor-1}$ weakly covers all edges of $P$.
It follows that $R_{\lfloor n/2\rfloor-1}$ does not have any window, otherwise Lemma~\ref{lem:saturated} would imply that $P$ has an edge that is not weakly covered. Therefore ${\rm int}(P)\subseteq R_{\lfloor n/2\rfloor-1}$, as claimed. Corollary~\ref{cor:rk-in-vk} then implies ${\rm int}(P) \subseteq V_{\lfloor n/2\rfloor-1} (s)$, proving the theorem.
\end{proof}

\section{Diffuse Reflection Paths between Boundary Points}
\label{sec:boundary}

In this section, we prove Theorem~\ref{thm:boundary}. If the light source $s$ is in the interior of the polygon $P$, then Theorem~\ref{thm:upper} and Proposition~\ref{pp:closure} readily imply $P\subseteq R_{\lfloor n/2\rfloor}\subseteq V_{\lfloor n/2\rfloor}(s)$. It remains to consider diffuse reflection paths between points $s,t\in \partial P$, that is, between points on the boundary of $P$. If $s$ is a vertex of $P$, then no other points on the edges of $P$ incident to $s$ are illuminated by $s$. For example, in a triangle $P$, a diffuse reflection path between two vertices requires two turns (Fig.~\ref{fig:triangle}, left), and consequently $\overline{D}(P)=2$.

\begin{figure}[hbt]
\centering
\includegraphics[width=\textwidth]{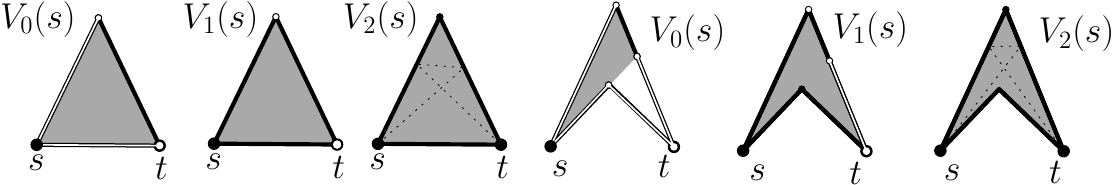}
\caption{The regions of a triangle (resp., a nonconvex quadrilateral) illuminated by the light source $s$ at the vertex after 0, 1, and 2 diffuse reflections.}
\label{fig:triangle}
\end{figure}

Let $s\in \partial P$, where $P$ is a simple polygon. We cannot use Proposition~\ref{pp:closure} when $s\in \partial P$. Proposition~\ref{pp:closure} is replaced by the following weaker statement.

\begin{proposition}
\label{pp:closure+boundary}
Let $s\in \partial P$ and $k\in \mathbb{N}_0$. Suppose that $U\subseteq V_k(s)$ such that ${\rm cl}(U)$ is a simple polygon. Then, $V_{k+1}(s)$ contains all points of $\partial U$ with the possible exception of the convex vertices of ${\rm cl}(U)$; and we have ${\rm cl}(U)\subseteq V_{k+2}(s)$.
\end{proposition}

\begin{proof}
Let $p\in \partial U\setminus U$. For every $i\in \mathbb{N}$, there is a point $p_i \in V_k(s)$ lying in a $\frac{1}{i}$-neighborhood of $p$ such that there is a diffuse reflection path $(s,r_i(1),\ldots , r_i(\ell),p_i)$ with $\ell\leq k$, where the points $r_i(1),\ldots , r_i(\ell)$ lie in the interior of some edges of $P$. By perturbing $p_i$, if necessary, we may assume that the ray $\overrightarrow{r_i(\ell)p_i}$ hits the boundary of $P$ at a point $q_i$ lying in the relative interior of an edge of $P$. If $i\in \mathbb{N}$ is sufficiently large, then $q_i$ directly sees $p$, unless $p$ and $q_i$ lie on the same edge of $P$, which means that $p$ is a convex vertex of ${\rm cl}(U)$. Consequently, if $p$ is not a convex vertex of ${\rm cl}(U)$, then $(s,r_i(1),\ldots , r_i(\ell),q_i,p)$ is a diffuse reflection path of length at most $k+1$ from $s$ to $p$.
If $p$ is a convex vertex of ${\rm cl}(U)$, then there is a point $r_i$ in the visibility polygon of $q_i$
that directly sees $p$, and so $(s,r_i(1),\ldots , r_i(\ell),q_i,r_i,p)$ is a diffuse reflection path of length at most $k+2$ from $s$ to $p$.
\end{proof}

We can define $R_k$ analogously to Section~\ref{sec:regions}. Let $R_0=V_0(s)$; and for $k\geq 1$,
let $R_{k}$ be the union of $R_{k-1}$, the sets $W_{ab}$ for all windows $ab$ of $R_{k-1}$, the boundary $\partial R_{k-1}$ with the exception of the convex vertices of ${\rm cl}(R_{k-1})$, and ${\rm cl}(R_{k-2})$ if $k\geq 2$. Proposition~\ref{pp:property} holds for $R_k$ for all $k\in \mathbb{N}_0$; and similarly to Corollary~\ref{cor:rk-in-vk}, we have $R_k\subseteq V_k(s)$ for all $k\in \mathbb{N}$.

Recall that $\mu_k$ is the number of edges of $P$ weakly covered by region $R_k$, and $\lambda_k$ is the number of windows of $R_k$. Instead of~\eqref{eq:1}, we maintain the following inequality for all $k \in \NN$:
\begin{equation}\label{eq:1+boundary}\tag{$\star\star$}
\mu_k\geq \min(2k+2,n).
\end{equation}
Inequality~\eqref{eq:1+boundary} combined with Proposition~\ref{pp:closure+boundary+} below, readily implies Theorem~\ref{thm:boundary}.

\begin{proposition}
\label{pp:closure+boundary+}
If $k\geq 1$ and $R_k$ weakly covers all edges of $P$ (i.e., $\mu_k=n)$, then $V_{k+1}(s)=P$.
\end{proposition}
\begin{proof}
Since $\mu_k=n$, then $R_k$ has no windows and ${\rm int}(P)\subseteq R_k$.
By Proposition~\ref{pp:closure+boundary}, $R_{k+1}$  contains $\partial P$ with the possible exception of the convex vertices of $P$.  As $R_k \subseteq R_{k+1}$ and $R_{k+1} \subseteq V_{k+1}(s)$,  it only remains to show that all convex vertices of $P$ are in $V_{k+1}(s)$. Consider a convex vertex $v$ of $P$. If $v\in \partial R_{k-1}$, then $v\in R_{k+1}$ by Proposition~\ref{pp:closure+boundary}. Suppose $v \in \partial R_k$ but $v\not\in \partial R_{k-1}$. Then $v$ is incident to some region $U_{ab}$ separated from $R_{k-1}$ by a window $ab$, where ${\rm int}(U_{ab})\subseteq W_{ab}$. If $U_{ab}$ is saturated, then all boundary points of $U_{ab}$ are in $R_k$, hence in $R_{k+1}$ and $V_{k+1}(s)$. If $U_{ab}$ is unsaturated, then ${\rm int}(U_{ab})$ is visible from a segment $c \in R_{k-1}$ described in property (iii). Since $v\not\in R_k$, vertex $v$ is incident to the edge of $P$ that contains $c$. In this case, however, there is a diffuse reflection path from $c$ to $v$ with one reflection, and as $c \in R_{k-1} \subseteq V_{k-1}(s)$, then $v\in V_{k+1}(s)$, as desired.
\end{proof}

We argue that \eqref{eq:1+boundary} holds for all $k\in \mathbb{N}$. Lemma~\ref{lem:saturated} holds when $s\in \partial P$, but some of the propositions in Section~\ref{sec:upper} require adjustments. Proposition~\ref{pp:ini} (i.e., $\mu_0\geq 3$) is replaced by the following:

\begin{proposition}\label{pp:ini+boundary}
If $s\in \partial P$, then $R_0=V_0(s)$ weakly covers at least one edge of $P$, and $R_1$ weakly covers at least $\min(n,4)$ edges of $P$.
\end{proposition}

\begin{proof}
As argued in the proof of Proposition~\ref{pp:ini}, the boundary of $R_0=V_0(s)$ contains line segments from at least three edges of $P$.
Hence ${\rm cl}(R_0)$ weakly covers at least three edges of $P$.
However,~$s$ cannot see any point in the edges of $P$ that contain~$s$.
At most two edges of $P$ contain $s$, hence $V_0(s)$ weakly covers at least one edge of $P$.

All interior points of the edges of $\partial R_0$ can be reached from $s$ after one diffuse reflection.
Hence the region $R_1$ covers at least three edges of $P$ that are weakly covered by ${\rm cl}(R_0)$.
This completes the proof for $n=3$. If $n\geq 4$, then either ${\rm cl}(R_0)=P$ and so $R_1$ weakly
covers all edges of $P$, or $R_0$ has a window and $R_1$ covers at least one edge behind the window by Lemma~\ref{lem:saturated}(1).
\end{proof}

By Proposition~\ref{pp:ini+boundary}, inequality~\eqref{eq:1+boundary} holds in the initial case $k=1$, i.e., $\mu_1\geq 4$. In this section, we consider a region $R_k$ \emph{critical} if $\mu_k=2k+2$ and $\mu_k<n$.
Conditions~(A) and~(B) can now be adapted verbatim.
Propositions~\ref{pp:ini1} and \ref{pp:ini2} are replaced by a single claim about $R_1$:

\begin{proposition}\label{pp:ini1+boundary}
If $R_1$ is critical, then
$R_1$ satisfies (A).
\end{proposition}

\begin{proof}
If $R_1$ is critical, then ${\rm cl}(R_1)\neq P$, and so $R_0$ and $R_1$ each have at least one window. Recall that ${\rm cl}(R_0)$ weakly covers at least three edges of $P$, and $R_1$ weakly covers at least one additional edge of $P$ that is not weakly covered by ${\rm cl}(R_0)$. Since $R_1$ is critical, we have $\mu_1=4$, hence ${\rm cl}(R_0)$ weakly covers precisely three edges of $P$, and $R_1$ weakly covers precisely one additional edge $ad$. By Lemma~\ref{lem:saturated}, $R_0$ has a unique unsaturated window, say $ab$. Refer to Fig.~\ref{fig:zigzags}.

Because the window of $R_0$ is unsaturated, $R_1={\rm cl}(R_0)\cup W_{ab}$, where $W_{ab}=V_0(c)\cap U_{ab}$ for a segment $c$ described in property (iii) lying in the relative interior of some edge $e$ of $P$. Since the region $R_1$ weakly covers only one new edge not weakly covered by ${\rm cl}(R_0)$, the supporting line of $e$ intersects $ad$ at some point $\beta$. Denote by $\alpha$ the endpoint of $e$ that lies in the segment $b\beta$. Observe that $W_{ab}=\Delta(ab\beta)\setminus b\beta$, and all windows of $R_1$ are contained in $\alpha\beta$, which is collinear with $e$. 
 Every chord of $P$ that crosses $\alpha\beta$ has an endpoint in either the relative interior of $a\beta$ or in ${\rm cl}(R_0) \setminus \{e\}$, which are contained in $R_1$. As desired, all windows of $R_1$ are saturated and $R_1$ satisfies (A).
\end{proof}

After replacing~\eqref{eq:1} with~\eqref{eq:1+boundary} and using the new definition of critical regions, Propositions~\ref{pp:crit}, \ref{pp:1}, \ref{pp:A}, and \ref{pp:prop-crit-three-step}, as well as Lemma~\ref{lem:induction}, go through, showing that \eqref{eq:1+boundary} is maintained for all $k\in \mathbb{N}$.
We are now ready to prove Theorem~\ref{thm:boundary}.

\begin{figure}[htb]
\centering
\includegraphics[width=0.9\textwidth]{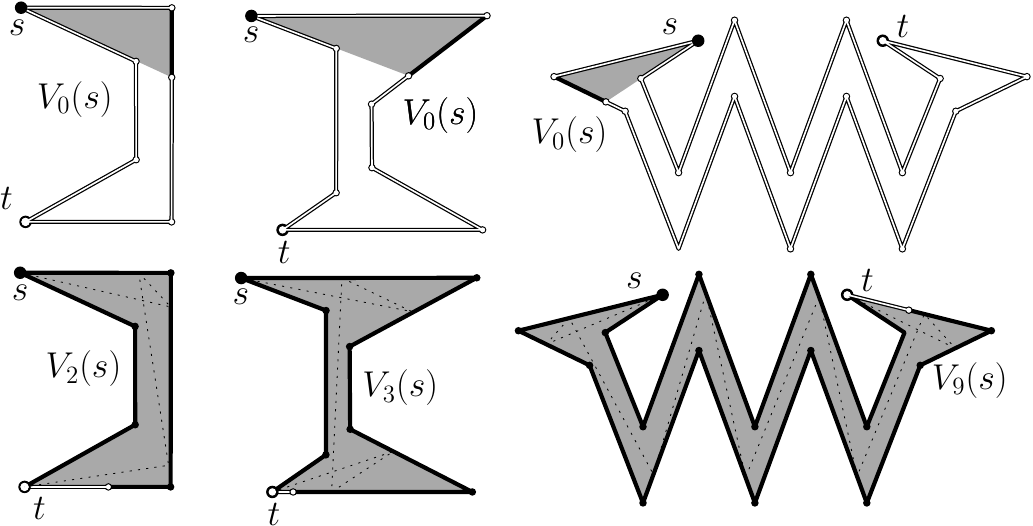}
\caption{Simple polygons with $n=6$, 8, and 18 vertices, where
every diffuse reflection path between $s$ and $t$ has at least
$\lfloor n/2 \rfloor$ turns.
Top row: the regions $V_0(s)$.
Bottom row: the regions $V_{\lfloor n/2\rfloor-1}(s)$, which
contain the interior ${\rm int}(P)$, and all points on the boundary
$\partial P$ except for a line segment incident to $t$.}
\label{fig:zigzags}
\end{figure}

\begin{customthm}{\ref{thm:boundary}}
We have $\overline{D}(3)=2$ and $\overline{D}(n)= \lfloor n/2\rfloor$ for every integer $n\geq 4$.
\end{customthm}

\begin{proof}
It is easily verified that $D(3)=2$ (see Fig.~\ref{fig:triangle}).
We show that in every simple polygon $P$ with $n\geq 4$ vertices,
there exists a diffuse reflection path with at most $\lfloor n/2\rfloor$
reflections between any two points $s,t\in P$. Theorem~\ref{thm:upper}
implies that it is enough to prove $P \subseteq V_k(s)$ for every $s\in \partial P$ and $k\geq \lfloor n/2\rfloor$. By~\eqref{eq:1+boundary}, $R_{\lceil n/2\rceil-1}$ weakly covers all
edges of $P$. It follows that region $R_{\lceil n/2\rceil-1}$ does not have any window,
and so ${\rm int}(P)\subseteq R_{\lceil n/2\rceil-1}$.

Suppose first that $n$ is even. Then ${\rm int}(P)\subseteq R_{\lceil n/2\rceil-1}$ combined with Proposition~\ref{pp:closure+boundary+} yields $P = V_{\lceil n/2\rceil}(s)$, hence $P=V_{\lfloor n/2\rfloor}(s)$, as required.
Suppose now that $n$ is odd and $n=2\ell+1$ for some $\ell>1$.
By~\eqref{eq:1+boundary}, $R_{\ell-1}$ weakly covers at least $2\ell=n-1$ edges of $P$.
If $R_{\ell-1}$ weakly covers all edges of $P$, then $P=V_\ell(s)=V_{\lfloor n/2\rfloor}(s)$ by Proposition~\ref{pp:closure+boundary+}. Otherwise $R_{\ell-1}$ weakly covers exactly $n-1$ edges of $P$.
Consequently, $R_{\ell-1}$ is critical, and we have $\lambda_{\ell-1}=1$ by Lemma~\ref{lem:saturated}(1).
By Lemma~\ref{lem:induction}(2), $R_{\ell-1}$ has a saturated window $ab$. In this case, either $U_{ab}=ab$ or $U_{ab}$ is a triangle adjacent to $ab$. Hence $U_{ab}\subseteq V_0(ab)$, and so $P=R_\ell=R_{\lfloor n/2\rfloor}\subseteq V_{\lfloor n/2\rfloor}(s)$, as claimed.

The matching lower bound $\overline{D}(n)\geq \lfloor n/2\rfloor$ for $n\geq 4$ follows from a family
construction. For every $n\geq 4$, there is a simple polygon $P_n$ with $n$ vertices, including $s,t\in P_n$, such that every diffuse reflection path between $s$ and $t$ has $\lfloor n/2\rfloor$ reflections. For odd $n$, $n\geq 5$, the polygon $P_n$ is obtained by subdividing an arbitrary edge of $P_{n-1}$. The polygon $P_4$ is a nonconvex quadrilateral, where $s$ and $t$ are two opposite convex corners (Fig.~\ref{fig:triangle}). Polygon $P_6$ is depicted in Fig.~\ref{fig:zigzags}. For even integers $n\geq 8$, the polygon $P_n$ is constructed by attaching two nonconvex quadrilaterals to a zigzag polygon as in Fig.~\ref{fig:zigzags}.
\end{proof}

\section{Conclusion}

We have shown that in every simple polygon with $n$ vertices, every point light source $s\in {\rm int}(P)$ can illuminate the interior of $P$ after at most $\lfloor n/2\rfloor -1$ diffuse reflections, and this bound is the best possible. A point light source $s\in P$, either in the interior or boundary of $P$, can illuminate $P$ after at most $\lfloor n/2\rfloor$ diffuse reflections for any $n \geq 4$, and this bound is tight. However, the diffuse reflection diameter may be significantly smaller for a given polygon $P$ (e.g., convex polygons).
Several problems related to diffuse reflection paths remain open:

\begin{itemize}\itemsep -1pt
\item
Is there an efficient algorithm for finding the diffuse reflection diameter of a given simple polygon $P$ with $n$ vertices? Combining our result with the bound $O(n^9)$ on the complexity of $V_k(s)$ by Aronov \etal~\cite{AD06}, we can compute in polynomial time the minimum $k\in \mathbb{N}_0$ such that $P=V_k(s)$ for any point $s\in P$. But it is unclear how many points $s\in P$ would have to be tested to find the maximum.
\item
Is there an efficient data structure for a simple polygon $P$ that, for a query point pair $s,t\in {\rm int}(P)$, would report a diffuse reflection path between $s$ and $t$ with the minimum number of reflections? Arkin \etal~\cite{AMS95} designed a data structure for analogous queries for minimum link paths in a simple polygon.
\item
What is the maximum diffuse reflection diameter of a star-shaped polygon $S$ with $n$ vertices?
Even though every point of $S$ is visible from some point $s\in S$, it is not clear how a diffuse reflection path could take advantage of this property.
Our lower bound constructions do not extend to star-shaped polygons.
\item
What is the maximum diffuse reflection diameter of a simple polygon with $n$ vertices, $r$ of which are reflex? It is clear that no reflection is necessary for $r=0$, but the dependence on the parameter $r$ is not clear.
\item
What is the maximum diffuse reflection diameter of a polygon with $h$ holes and a total of $n$ vertices?
\end{itemize}

\section*{Acknowledgements}

We would like to thank all reviewers for helpful comments improving the presentation of the paper and the MIT-Tufts Research Group on Computational Geometry for initial discussions of the problem.



\begin{thebibliography}{10}

\bibitem{ABP08}
M.~Aanjaneya, A.~Bishnu, S.~P. Pal, Directly visible pairs and illumination by
  reflections in orthogonal polygons, in: Proceedings of 24th European Workshop
  on Computational Geometry, 2008, pp. 241--244.

\bibitem{AMS95}
E.~Arkin, J.~Mitchell, S.~Suri, Logarithmic-time link path queries in a simple
  polygon, Internat. J. Comput. Geom. Appl. 5~(4) (1995) 75--79.

\bibitem{AD+98}
B.~Aronov, A.~Davis, T.~K. Dey, S.~P. Pal, D.~C. Prasad, Visibility with one
  reflection, Discrete \& Computational Geometry 19~(4) (1998) 553--574.

\bibitem{ADD+98}
B.~Aronov, A.~Davis, T.~K. Dey, S.~P. Pal, D.~C. Prasad, Visibility with
  multiple reflections, Discrete \& Computational Geometry 20~(1) (1998)
  61--78.

\bibitem{AD06}
B.~Aronov, A.~Davis, J.~Iacono, A.~S.~C. Yu, The complexity of diffuse
  reflections in a simple polygon, in: J.~R. Correa, A.~Hevia, M.~Kiwi (Eds.),
  LATIN 2006: Theoretical Informatics, Vol. 3887 of LNCS, Springer Berlin
  Heidelberg, 2006, pp. 93--104.

\bibitem{AGS00}
T.~Asano, S.~K.~Ghosh, T.~Shermer, Visibility in the plane,
in \emph{Handbook of Computational Geometry (J.-R. Sack and J. Urrutia, eds.)},
North-Holland, Amsterdam, 2000, pp.~829--876.

\bibitem{BPS04}
S.~Brahma, S.~P. Pal, D.~Sarkar, A linear worst-case lower bound on the number
  of holes inside regions visible due to multiple diffuse reflections, Journal
  of Geometry 81~(1--2) (2004) 5--14.

\bibitem{G07}
S.~K. Ghosh, Visibility algorithms in the plane, Cambridge Univ. Press, 2007,
  Ch.~7, pp. 218--254.

\bibitem{GG+09}
S.~K. Ghosh, P.~P. Goswami, A.~Maheshwari, S.~C. Nandy, S.~P. Pal,
  S.~Sarvattomananda, Algorithms for computing diffuse reflection paths in
  polygons, The Visual Computer 28~(12) (2012) 1229--1237.

\bibitem{K69}
V.~Klee, Is every polygon illuminable from some point?, American Mathematical
  Monthly 76 (1969) 180.

\bibitem{MSD00}
A.~Maheshwari, J.-R. Sack, H.~N. Djidjev, Link distance problems, in: Handbook
  of Computational Geometry, Elsevier, 2000, Ch.~12, pp. 519--558.

\bibitem{PPD98}
D.~C. Prasad, S.~P. Pal, T.~K. Dey, Visibility with multiple diffuse
  reflections, Computational Geometry 10~(3) (1998) 187--196.

\bibitem{S86}
S.~Suri, A linear time algorithm for minimum link paths inside a simple
  polygon, Comput. Vision. Graph. Image Process. 35 (1986) 99--110.

\bibitem{S90}
S.~Suri, On some link distance problems in a simple polygon, IEEE Trans. Robot.
  Autom. 6 (1990) 108--113.

\bibitem{T95}
G.~Tokarsky, Polygonal rooms not illuminable from every point, The American
  Mathematical Monthly 102~(10) (1995) 867--879.



\end{thebibliography}
\end{document}